\documentclass[11pt]{amsart}

\usepackage[utf8]{inputenc}
\usepackage[T1]{fontenc}
\usepackage{lmodern}
\usepackage{amsmath,amssymb,amsthm,mathtools,mathrsfs}
\usepackage[bookmarksnumbered=true,colorlinks=true,linkcolor=blue,citecolor=blue,urlcolor=blue]{hyperref}
\usepackage{comment}

\newtheorem{definition}{Definition}[section]
\newtheorem{lemma}[definition]{Lemma}
\newtheorem{proposition}[definition]{Proposition}
\newtheorem{theorem}[definition]{Theorem}
\newtheorem{corollary}[definition]{Corollary}
\theoremstyle{definition}

\numberwithin{equation}{section}
\allowdisplaybreaks

\newcommand{\R}{\mathbb R}
\newcommand{\C}{\mathbb C}
\newcommand{\Z}{\mathbb Z}
\newcommand{\dpm}{\frac{dp}{(2\pi)^3}}
\newcommand{\1}{\mathbf 1}

\title[Derivation of the BCS functional from the BCS Hamiltonian]
{A simple derivation of the zero-temperature BCS functional from the reduced BCS Hamiltonian}
\author[C. Hainzl]{Christian Hainzl}
\address{Mathematisches Institut, Ludwig-Maximilians-Universit\"at M\"unchen,
Germany}
\email{hainzl@math.lmu.de}
\author[R. Panza]{Riccardo Panza}
\address{Mathematisches Institut, Ludwig-Maximilians-Universit\"at M\"unchen,
Germany}
\email{panza@math.lmu.de}

\date{}

\dedicatory{Dedicated to Robert Seiringer on the occasion of his 50th birthday}

\begin{document}

\begin{abstract}
We consider a grand-canonical reduced BCS Hamiltonian on the full fermionic Fock space with an attractive pairing interaction supported in an energy shell around the Fermi surface. For fixed chemical potential, we prove that the exact finite-volume ground-state energy differs from the minimum of the corresponding finite-volume BCS functional by a volume-independent error. As a consequence, the BCS variational principle is exact for the ground-state energy density in the thermodynamic limit.
\end{abstract}

\maketitle

\section{Introduction}\label{sec:introduction}

The reduced BCS Hamiltonian is one of the simplest effective many-body models in which the Cooper-pairing mechanism can be studied directly \cite{BCS1957}. We work with spin-$\frac{1}{2}$ fermions in a periodic box \(\Lambda_L=[-L/2,L/2]^3\), at chemical potential \(\mu>1\),
and take the pairing interaction to be supported on the momenta in
\[
 \Omega_{\mu}
 :=
 \{p\in\R^3:\bigl||p|^2-\mu\bigr|\le1\}.
\]
Even though our proof works for a more general rank-one interaction, we stick to the simpler physically relevant case (for the physics derivation of such a nonlocal interaction, we refer to standard textbooks such as \cite{Fetter}).\\
In this setting, the reduced BCS Hamiltonian reads as follows, as an operator on the fermionic Fock space $\mathcal{F}_f(L^2(\Lambda_L; \C^2))$:
\begin{equation}\label{eq:ReducedHam}
    H_{L, \mu} := \sum_{p \in \Lambda_L^* \atop \sigma \in \{\uparrow, \downarrow\}}(|p|^2 - \mu)a^*_{p, \sigma}a_{p,\sigma} - \frac{g}{\sqrt{\mu}L^3}\sum_{p,q \in \Lambda_L^* \cap \Omega_{\mu}}b^*_p b_q\;,
\end{equation}
where $a^*_{p,\sigma}, a_{p, \sigma}$ are the usual creation and annihilation operators for plane waves of momentum $p \in \Lambda_L^* := \frac{2\pi}{L}\Z^3$ and spin $\sigma \in \{\uparrow, \downarrow\}$, satisfying the canonical anticommutation relations (CAR) $\{a_{p, \sigma}, a^*_{q, \tau}\} = \delta_{p,q}\delta_{\sigma, \tau}$. Here $b^*_p := a^*_{p, \uparrow} a^*_{-p, \downarrow}$ creates a pair of fermions with opposite momenta and spin.
\\
Moreover, \(g>0\) is a fixed coupling constant. The scaling \(\mu^{-1/2}\) is chosen to compensate, in the high-density limit, for the growth of the number of momentum states participating in the interaction inside the shell \(\Omega_\mu\). With this choice, the effective pairing strength remains of order one and, in particular, the BCS gap stays of order one as \(\mu\to\infty\).

We are interested in the ground-state energy
\begin{equation}
    E_0(L, \mu) := \inf_{\psi \in \mathcal{F}_f}\frac{\langle \psi, H_{L, \mu} \psi \rangle}{\langle \psi, \psi \rangle}\;,
\end{equation}
and in its relation, in the thermodynamic limit \(L\to\infty\), to the minimum of the BCS functional defined in \eqref{eq:BCS-functional}.

\subsection{Setting and main result}
To formulate our main result, we first introduce the finite- and infinite-volume BCS energy functionals associated with the separable interaction considered above. Their basic properties, including the characterization of their minimizers and the thermodynamic limit, are discussed in Section \ref{sec:functionals}. We begin with the standard zero-temperature BCS energy functional, defined for $\alpha: \R^3 \to \C$ satisfying $|\alpha(p)| \leq \frac{1}{2}$ by
\begin{equation}\label{eq:BCS-functional}
    \mathcal{F}^{BCS}_{\mu}[\alpha] := \int_{\mathbb R^3}|\,|p|^2-\mu\,|
 \left(1-\sqrt{1-4|\alpha(p)|^2}\right)\frac{dp}{(2\pi)^3} - \frac{g}{\sqrt{\mu}}\left|\int_{\Omega_{\mu}}\alpha(p)\,\frac{dp}{(2\pi)^3} \right|^2\;,
\end{equation}
and denote its minimum by
\begin{equation}
    e_{\rm BCS}(\mu) := \inf_{\alpha} \mathcal{F}^{\rm BCS}_{\mu}[\alpha]\;.
\end{equation}
The Euler--Lagrange equations associated with this variational problem can be written in the following self-consistent form, where $\chi_{\Omega_{\mu}}$ denotes the characteristic function of $\Omega_{\mu}$:
\begin{equation}\label{eq:minimizers-form}
    \alpha(p) = \frac{\Delta}{2\sqrt{|\,|p|^2-\mu\,|^2 + |\Delta|^2}}\chi_{\Omega_{\mu}}(p)\;, \qquad \Delta = \frac{g}{\sqrt{\mu}}\int_{\Omega_{\mu}}\alpha(p)\,\frac{dp}{(2\pi)^3}\,.
\end{equation}
Since the interaction is separable, the variational problem can equivalently be reduced to a scalar minimization over the complex gap parameter $\Delta$:
\begin{equation}\label{eq:intro-shell-BCS}
 e_{\rm BCS}(\mu)
 =\inf_{\Delta\in\mathbb C}
 \left\{
 \frac{\sqrt\mu}{g}|\Delta|^2
 +\int_{\Omega_\mu}
 \left(|\,|p|^2-\mu\,|-\sqrt{(|p|^2-\mu)^2+|\Delta|^2}\right)
 \frac{dp}{(2\pi)^3}
 \right\}.
\end{equation}
Since our main result concerns the finite-volume many-body Hamiltonian, we also introduce the corresponding finite-volume scalar BCS energy. Let $\Omega_{L, \mu}:= \Lambda_L^* \cap \Omega_{\mu}$ and define $\mathscr {E}_{L,\mu}^{BCS}: \C \to \R$ by
\begin{equation}\label{eq:BCS-scalar-finite-vol}
    \mathscr{E}_{L,\mu}^{\rm BCS}(\Delta) := \frac{\sqrt{\mu}L^3}{g}|\Delta|^2 + \sum_{p \in \Omega_{L, \mu}}\left(||p|^2 - \mu| - \sqrt{||p|^2 - \mu|^2 + |\Delta|^2} \right)\;,
\end{equation}
We denote the corresponding minimum by $E_{L, \mu}^{BCS} = \min_{\Delta}\mathscr{E}_{L,\mu}^{\rm BCS}(\Delta)$.\\
Finally, we denote the energy of the free Fermi sea by
\begin{equation}
   E_{\rm FS}(L, \mu) := 2 \sum_{p \in \Lambda_L^*\,: \atop |p|^2 < \mu} (|p|^2 - \mu)\;\,.
\end{equation}
We can now state our main finite-volume estimate.
\begin{theorem}[Finite-volume shell estimate]\label{thm:main}
There exist a sidelength $L_0 > 0$ and a constant $C_{\mu, g}$ such that, for every $L \geq L_0$,
\begin{equation}\label{eq:main-general}
 E_{\rm FS}(L,\mu)+E^{\rm BCS}_{L,\mu} - C_{\mu, g}\le E_0(L,\mu)
 \le E_{\rm FS}(L,\mu)+E^{\rm BCS}_{L,\mu}.
\end{equation}
The constant $C_{\mu, g}$ depends only on $\mu > 1$ and $g > 0$ and, in particular, is independent of the sidelength $L$.
\end{theorem}
Thus, the finite-volume BCS energy approximates the exact many-body ground-state energy up to an error that remains bounded uniformly in the volume. Together with the thermodynamic limit of the finite-volume BCS functional, this immediately yields the following corollary.
\begin{corollary}[Thermodynamic limit]\label{cor:thermodynamic}
For every fixed \(\mu>1\) and \(g>0\),
\begin{equation}\label{eq:thermodynamic-limit}
 \lim_{L\to\infty}\frac{E_0(L,\mu)}{L^3}
 =
 e_{\rm FS}(\mu)+e_{\rm BCS}(\mu),
 \quad
 e_{\rm FS}(\mu)
 :=2\int_{\{|p|^2 < \mu\}}\left(|p|^2 - \mu \right)\frac{dp}{(2\pi)^3}\,.
\end{equation}
\end{corollary}
We finally consider the high-density regime $\mu \to \infty$. The scaling of the coupling chosen above leads to a nontrivial BCS correction to the ground-state energy density, while the corresponding gap parameter remains of order one. More precisely, we obtain the following asymptotics.
\begin{corollary}[High-density asymptotics]\label{cor:high-density}
Define
\[
 e_0(\mu)=\lim_{L\to\infty}\frac{E_0(L,\mu)}{L^3}
\]
and, for fixed \(g>0\),
\begin{equation}\label{eq:Delta-infty}
 \Delta_\infty=\frac1{\sinh(4\pi^2/g)}.
\end{equation}
Then, as \(\mu\to\infty\),
\begin{equation}
    e_{\rm BCS}(\mu) =\frac{\sqrt\mu}{4\pi^2}
 \left(1-\sqrt{1+\Delta_\infty^2}\right)
 +O(\mu^{-3/2})\;,
\end{equation}
and hence
\begin{equation}
    e_0(\mu) = -\frac{2}{15\pi^2}\mu^{5/2}
 +\frac{\sqrt\mu}{4\pi^2}
 \left(1-\sqrt{1+\Delta_\infty^2}\right)
 +O(\mu^{-3/2})\,.
\end{equation}
Moreover, the minimizer of the BCS functional \ref{eq:BCS-functional} is of the form \eqref{eq:minimizers-form}, with gap parameter $\Delta_{\mu}$ satisfying
\begin{equation}\label{eq:gap-high-density}
    \Delta_{\mu} = \Delta_{\infty} + O(\mu^{-2})\,.
\end{equation}
\end{corollary}

\subsection{Previous works}

The BCS functional itself is by now a well-developed effective theory. In particular,
Ginzburg--Landau theory was derived rigorously from the BCS functional close to the
critical temperature in \cite{FrankHainzlSeiringerSolovej2012}, with subsequent extensions
to magnetic and more general external fields in
\cite{DeuchertHainzlMaier2023Homogeneous,DeuchertHainzlMaier2023General}. A broader and
more difficult objective is to understand how BCS theory emerges from a less reduced
microscopic many-fermion problem. Renormalization-group approaches identify the Cooper
channel as the distinguished infrared instability of weakly interacting fermions near the
Fermi surface; see, for example,
\cite{FeldmanTrubowitz1990,FeldmanTrubowitz1991,FeldmanMagnenRivasseauTrubowitz1992,
FeldmanMagnenRivasseauTrubowitz1993,ChenFroehlichSeifert1996}. We will not pursue this
microscopic direction here.

For the reduced BCS Hamiltonian itself, the relation to the Bogoliubov--BCS description has a long history. Particularly relevant to the present work is the method of the approximating Hamiltonian, originating in the early work of Bogolyubov and subsequently developed in a much broader framework; see, for example, the survey \cite{BogolRev} or the book \cite{BogolyubovJr1972}. For an attractive separable interaction, the collective pair field is replaced by a complex parameter, leading to a quadratic Hamiltonian together with a residual fluctuation term measuring the deviation of the collective field from this parameter. This provides the basic algebraic mechanism underlying the BCS mean-field description.

A mathematically rigorous analysis of the thermodynamic limit was obtained later by several methods. In particular, Duffield and Pul'e used large-deviation techniques together with Berezin--Lieb inequalities to derive the limiting free-energy variational principle for BCS-type quasi-spin models \cite{DuffieldPule1987,DuffieldPule1988}. Related rigorous mean-field variational principles were developed in \cite{CeglaLewisRaggio1988,PetzRaggioVerbeure1989,RaggioWerner1989BCS,RaggioWerner1991}.

A complementary infinite-volume formulation was developed by Haag \cite{Haag1962}. In an irreducible representation of the observable algebra, the spatially averaged pair field becomes a c-number, the quartic pairing interaction is represented by a quadratic Bogoliubov Hamiltonian, and the gap equation arises as a self-consistency condition. Thirring and Wehrl \cite{ThirringWehrl1967} subsequently gave a precise infinite-tensor-product formulation of this mechanism and studied the convergence of the dynamics of local observables to the Bogoliubov dynamics. Thirring \cite{Thirring1968} further analyzed thermal Green functions for the degenerate quasi-spin model.

Other works address complementary aspects of the reduced model. Richardson and Richardson--Sherman \cite{Richardson1963,RichardsonSherman1964} constructed exact finite-volume eigenstates, while spectral properties of separable reduced Hamiltonians were studied in \cite{Kato1965,Kato1967}. Mattis and Lieb \cite{MattisLieb1961} obtained explicit exact wave functions and analyzed their large-volume relation to the BCS variational equations, including finite-size corrections. Bursill and Thompson \cite{BursillThompson1993} studied the validity of the BCS variational description for classes of attractive reduced pairing interactions.

The classical literature therefore contains both the physical approximating-Hamiltonian mechanism and rigorous thermodynamic-limit results for reduced BCS-type models. Our purpose here is different. We work directly with the finite-volume fermionic Hamiltonian and obtain an explicit, volume-independent bound on the difference between its exact ground-state energy and the finite-volume BCS variational energy. In this sense, the result provides a quantitative finite-volume version of the thermodynamic exactness of the BCS approximation for the shell model considered here.

\section{Elementary properties of BCS functionals}
\label{sec:functionals}
In this section, we collect the basic properties of the finite- and infinite-volume BCS functionals that will be used throughout the paper. Although at zero temperature the BCS variational problem can be reduced to a functional of the pairing function $\alpha$ alone, it is convenient for our purposes to retain both the one-particle density $\gamma$ and the pairing density $\alpha$. This formulation arises naturally in the BCS product-state upper bound of Section \ref{sec:upper}.\\
\textbf{Finite volume.}
For each \(p\in\Omega_{L,\mu}\), let \(\gamma_p\in[0,1]\) and \(\alpha_p\in\C\) satisfy
\[
 |\alpha_p|^2 = \gamma_p(1-\gamma_p)\,.
\]
After subtracting the energy of the free Fermi sea, the finite-volume BCS functional is given by
\begin{equation}\label{eq:finite-gamma-alpha}
 \mathscr F^{BCS}_{L,\mu}[\gamma,\alpha]
 :=
 2\sum_{p\in\Omega_{L,\mu}}
 \left(|p|^2 - \mu \right)\bigl(\gamma_p-\1_{\{|p|^2 < \mu\}}\bigr)
 -\frac{g}{\sqrt\mu L^3}
 \left|\sum_{p\in\Omega_{L,\mu}}\alpha_p\right|^2.
\end{equation}

As explained in the introduction, the separable structure of the interaction allows this variational problem to be reduced to a scalar one. The corresponding finite-volume scalar functional was defined in \eqref{eq:BCS-scalar-finite-vol}
\begin{equation*}
    \mathscr{E}_{L,\mu}^{\rm BCS}(\Delta):= \frac{\sqrt{\mu}L^3}{g}|\Delta|^2 + \sum_{p \in \Omega_{L, \mu}}\left(||p|^2 - \mu| - \sqrt{||p|^2 - \mu|^2 + |\Delta|^2} \right)\,.
\end{equation*}

\textbf{Infinite volume.}
In infinite volume, we consider measurable functions $(\gamma,\alpha):\Omega_\mu\to[0,1]\times\C$ satisfying $|\alpha|^2=\gamma(1-\gamma)$ almost everywhere. Since the interaction is restricted to the shell $\Omega_\mu$, which has finite measure, all terms in the following functional are finite:
\begin{align}
 \mathscr F^{\rm BCS}_{\mu}[\gamma,\alpha]
 :={}&
 2\int_{\Omega_\mu}
 (|p|^2-\mu)
 \bigl(\gamma(p)-\1_{\{|p|^2<\mu\}}\bigr)\,\dpm
 \nonumber\\
 &-\frac{g}{\sqrt\mu}
 \left|\int_{\Omega_\mu}\alpha(p)\,\dpm\right|^2
 \label{eq:infinite-gamma-alpha}
\end{align}
The corresponding scalar functional, obtained from the same reduction as in finite volume, is
\begin{equation}
  \mathscr E_{\mu}^{\rm BCS}(\Delta) := \frac{\sqrt\mu}{g}|\Delta|^2
 +\int_{\Omega_\mu}
 \left(|\,|p|^2-\mu\,|-\sqrt{(|p|^2-\mu)^2+|\Delta|^2}\right)
 \frac{dp}{(2\pi)^3}\,.
\end{equation}

The relation between the $(\gamma,\alpha)$ formulation and the scalar formulation, together with the relevant thermodynamic convergence, is summarized in the following proposition.
\begin{proposition}[Scalar reduction]
\label{prop:scalar-reduction}
There exist numbers
\[
 \Delta_{L,\mu}\geq0,
 \qquad
 \Delta_\mu>0,
\]
such that
\begin{align}
 E^{\rm BCS}_{L,\mu}
 &:=
 \min_{(\gamma,\alpha)}
 \mathscr F^{\rm BCS}_{L,\mu}[\gamma,\alpha]
 =
 \min_{\Delta} \mathscr E^{\rm BCS}_{L,\mu}(\Delta) = \mathscr E^{\rm BCS}_{L,\mu}(\Delta_{L, \mu}),
 \label{eq:scalar-reduction-finite}
 \\
 e_{\rm BCS}(\mu)
 &:=
 \min_{(\gamma,\alpha)}
 \mathscr F^{\rm BCS}_{\mu}[\gamma,\alpha]
 =
 \min_{\Delta} \mathscr E^{\rm BCS}_{\mu}(\Delta) = \mathscr E^{\rm BCS}_{\mu}(\Delta_{\mu})
 \label{eq:scalar-reduction-infinite}
\end{align}
Moreover, for all sufficiently large $L$, the finite-volume gap parameter satisfies $\Delta_{L,\mu}>0$, and every minimizer of the finite-volume BCS functional is, up to a global phase, of the form 
\begin{equation}\label{eq:form-finite-minimizers}
    \gamma_p = \frac{1}{2}\left(
 1-\frac{|p|^2 - \mu}
 {\sqrt{||p|^2 - \mu|^2+\Delta_{L,\mu}^2}}
 \right),
 \qquad
 \alpha_p
 =
 \frac{e^{i\theta}\Delta_{L,\mu}}
 {2\sqrt{||p|^2 - \mu|^2+\Delta_{L,\mu}^2}},
\end{equation}
for some $\theta \in [0,2\pi)$. The infinite-volume minimizers have the analogous form with $\Delta_{L,\mu}$ replaced by $\Delta_\mu$.\\
The finite- and infinite-volume gap parameters are characterized, respectively, by the following gap equations:
\begin{align}
  1 &=
 \frac{g}{2\sqrt\mu\,L^3}
 \sum_{p\in\Omega_{L,\mu}}
 \frac1{\sqrt{||p|^2 - \mu|^2+\Delta_{L,\mu}^2}}\,, \\
  1 &=
 \frac{g}{2\sqrt\mu}
 \int_{\Omega_\mu}
 \frac1{\sqrt{||p|^2-\mu|^2+\Delta_\mu^2}}\frac{dp}{(2\pi)^3}\,.
\end{align}
Finally, the finite-volume minimum and the corresponding gap parameter converge to their infinite-volume counterparts as $L\to\infty$:
\begin{equation}
     \lim_{L \to \infty}\frac{E^{\rm BCS}_{L, \mu}}{L^3} = e_{\rm BCS}(\mu)\,, \qquad \lim_{L \to \infty}\Delta_{L, \mu} = \Delta_{\mu}\,.
\end{equation}
\begin{proof}
  This is a standard computation, see e.g. \cite{martin2013many}. For the sake of completeness, we sketch the proof of \eqref{eq:scalar-reduction-finite}.\\
  By a simple completion of the square, we can write the second term in $\mathscr{F}^{\rm BCS}_{L, \mu}[\gamma, \alpha]$ as
\begin{align*}
-\frac{g}{\sqrt\mu L^3}
\left|\sum_{p\in\Omega_{L,\mu}}\alpha_p\right|^2
=
\min_{\Delta\in\C}
\left\{
\frac{\sqrt\mu L^3}{g}|\Delta|^2
-2\Re\left(
\overline{\Delta}
\sum_{p\in\Omega_{L,\mu}}\alpha_p
\right)
\right\}.
\end{align*}
Since the completion-of-the-square identity holds pointwise in \((\gamma,\alpha)\), and the finite-volume variational problem is finite-dimensional, we may interchange the two minimizations.
 Thus, for every
$p\in\Omega_{L,\mu}$, we want to find
\[
\min_{\substack{0\leq\gamma_p\leq1\\
|\alpha_p|^2=\gamma_p(1-\gamma_p)}}
\left\{
2(|p|^2-\mu)
\left(\gamma_p-\1_{\{|p|^2<\mu\}}\right)
-2\Re(\overline{\Delta}\alpha_p)
\right\}.
\]
By Cauchy--Schwarz on $\C^2$, together with the constraint on $(\alpha, \gamma)$, we find
\begin{align*}
&(|p|^2-\mu)(2\gamma_p-1)
-2\Re(\overline{\Delta}\alpha_p)
\\
&\qquad\geq
-\sqrt{(|p|^2-\mu)^2+|\Delta|^2} \underbrace{\sqrt{(2\gamma_p -1)^2 + 4|\alpha_p|^2}}_{= 1}.
\end{align*}
Using
\[
(|p|^2-\mu)
-2(|p|^2-\mu)\1_{\{|p|^2<\mu\}}
=
\bigl||p|^2-\mu\bigr|,
\]
we obtain
\begin{align*}
&\min_{\substack{0\leq\gamma_p\leq1\\
|\alpha_p|^2=\gamma_p(1-\gamma_p)}}
\left\{
2(|p|^2-\mu)
\left(\gamma_p-\1_{\{|p|^2<\mu\}}\right)
-2\Re(\overline{\Delta}\alpha_p)
\right\}
\\
&\qquad=
\bigl||p|^2-\mu\bigr|
-\sqrt{(|p|^2-\mu)^2+|\Delta|^2}.
\end{align*}
Summing over $p$ proves
\[
\min_{(\gamma,\alpha)}
\mathscr F^{\rm BCS}_{L,\mu}[\gamma,\alpha]
=
\min_{\Delta\in\C}
\mathscr E^{\rm BCS}_{L,\mu}(\Delta).
\]
The form of the minimizers in \eqref{eq:form-finite-minimizers} comes from the case of equality in Cauchy-Schwarz, while the existence of $\Delta_{L, \mu} \geq 0$ follows from convexity and coercivity of $\mathscr{E}^{\rm BCS}_{L, \mu}$, when considered as a function of $|\Delta|^2$.
\end{proof}
\end{proposition}

\section{Upper bound via BCS product states}
\label{sec:upper}
In this section, we prove the upper bound in Theorem \ref{thm:main} by evaluating the Hamiltonian in the standard BCS product state. This trial state goes back to the original work of Bardeen, Cooper, and Schrieffer \cite{BCS1957}; we recall the short argument for completeness.
Let $\Omega$ denote the Fock vacuum. We first fill all one-particle states below the Fermi surface that lie outside the interacting shell and define
\begin{equation}\label{eq:FS-out}
 \Psi_{\rm FS}^{\rm out}
 :=
 \prod_{\substack{p\notin\Omega_{L,\mu}\atop |p|^2 < \mu}}
 b_p^*\,\Omega.
\end{equation}
Inside the interacting shell, we use the usual BCS superposition of an empty and an occupied pair state. For each $p \in \Omega_{L, \mu}$, choose complex coefficients satisfying $(|u_p|^2+|v_p|^2=1)$ and set
\begin{equation}\label{eq:BCS-product}
 \Psi[u,v]
 :=
 \prod_{p\in\Omega_{L,\mu}}
 (u_p+v_pb_p^*)\,\Psi_{\rm FS}^{\rm out}.
\end{equation}
The canonical anticommutation relations imply that $\Psi[u,v]$ is normalized. Its energy expectation can be expressed exactly in terms of the corresponding one-particle and pairing densities. This yields the desired upper bound, as summarized in the following lemma.
\begin{lemma}[Product-state energy]
\label{lem:product-energy}
Given a product state as in \eqref{eq:BCS-product}, let
\begin{equation}\label{eq:product-gamma-alpha}
 \gamma_p :=|v_p|^2,\qquad
 \alpha_p :=\overline{u_p}\,v_p.
\end{equation}
Then
\begin{equation}\label{eq:exact-product-functional}
 \langle\Psi[u,v],H_{L,\mu}\Psi[u,v]\rangle
 =
 E_{\rm FS}(L,\mu)
 +\mathscr F^{\rm BCS}_{L,\mu}[\gamma,\alpha]
 -\frac{g}{\sqrt{\mu}L^3}\sum_{p\in\Omega_{L,\mu}}\gamma_p^2.
\end{equation}
Consequently,
\begin{equation}\label{eq:upper-bound}
 E_0(L,\mu) \le E_{\rm FS}(L,\mu) + E^{\rm BCS}_{L,\mu}.
\end{equation}

\begin{proof}
Using the product structure of $\Psi[u,v]$ together with the canonical anticommutation relations, one obtains, for $p,q\in\Omega_{L,\mu}$,
\begin{align*}
 &\langle \Psi[u,v], a^*_{p,\uparrow}a_{p,\uparrow}\, \Psi[u,v]\rangle
 =\langle \Psi[u,v], a^*_{-p,\downarrow}a_{-p,\downarrow}\, \Psi[u,v]\rangle =\gamma_p, \nonumber\\
 &\langle \Psi[u,v], b_p\, \Psi[u,v]\rangle = \alpha_p, \qquad \langle \Psi[u,v], b_p^*b_q \Psi[u,v]\rangle
 =\overline{\alpha_p}\alpha_q
 +\delta_{pq}\gamma_p^2.
\end{align*}
Substituting these identities into the Hamiltonian gives
\begin{align*}
    \langle \Psi[u,v], H_{L, \mu}\, \Psi[u,v]\rangle =&\, 2\sum_{p \notin \Omega_{L, \mu}\atop |p|^2 < \mu}(|p|^2 - \mu) + 2\sum_{p \in \Omega_{L, \mu}}(|p|^2 - \mu)\gamma_p\, + \\
    & -\frac{g}{\sqrt{\mu}L^3}\sum_{p,q \in \Omega_{L, \mu}}\overline{\alpha}_p\,\alpha_q - \frac{g}{\sqrt{\mu}L^3}\sum_{p \in \Omega_{L, \mu}}\gamma_p^2\;,
\end{align*}
To identify the BCS functional, we add and subtract the free Fermi-sea contribution from the interacting shell,
\begin{equation*}
    2\sum_{p \in \Omega_{L, \mu} \atop |p|^2 < \mu}(|p|^2 - \mu)
\end{equation*}
which gives \eqref{eq:exact-product-functional}. The upper bound \eqref{eq:upper-bound} then follows by minimizing over the admissible coefficients $u_p,v_p$ and using the variational principle for the many-body ground-state energy.
\end{proof}
\end{lemma}

\section{Approximating Hamiltonian and localization}
\label{sec:approximating}
In this section, we begin the proof of the lower bound in Theorem \ref{thm:main}. The argument combines the method of the approximating Hamiltonian \cite{BogolRev} with an IMS-type localization of the collective pair field.

The natural collective variable is the rescaled pair field associated with the interaction.

Let
\begin{equation}\label{eq:Z-def}
 Z=\frac{g}{\sqrt{\mu}L^3}\sum_{p \in \Omega_{L, \mu}}b_p\,.
\end{equation}
For every \(\omega\in\C\), the interaction can be decomposed exactly by completing the square around the complex parameter $\omega$. 
\begin{equation}\label{eq:square-completion}
 H_{L,\mu}
 =
 H_{\rm out}+H_{\rm app}(\omega)
 -\frac{\sqrt{\mu}L^3}{g}
 (Z-\omega)^*(Z-\omega),
\end{equation}
where
\begin{equation}\label{eq:H-out}
 H_{\rm out}
 =
 \sum_{p\notin\Omega_{L,\mu}}
 \left(|p|^2 - \mu \right)\bigl(n_{p,\uparrow}+n_{-p,\downarrow}\bigr)
\end{equation}
and
\begin{align}
 H_{\rm app}(\omega)
 &=
 \frac{\sqrt{\mu}L^3}{g}|\omega|^2
 +\sum_{p\in\Omega_{L,\mu}}h_p(\omega),
 \label{eq:H-app}\\
 h_p(\omega)
 &=
 \left(|p|^2 - \mu \right)\bigl(n_{p,\uparrow}+n_{-p,\downarrow}\bigr)
 -\overline\omega b_p-\omega b_p^*.
 \label{eq:h-p}
\end{align}
The advantage of the approximating Hamiltonian $H_{\rm app}(\omega)$ is that it decouples into independent momentum sectors. More precisely, for each $p\in\Omega_{L,\mu}$, the operator $h_p(\omega)$ acts on the four-dimensional local Fock space spanned by
\begin{equation*}
    |\Omega\rangle, \quad a^*_{p, \uparrow}|\Omega\rangle, \quad a^*_{-p, \downarrow}|\Omega\rangle, \quad b^*_p|\Omega\rangle = a^*_{p, \uparrow} a^*_{-p, \downarrow}|\Omega\rangle\;.
\end{equation*} 
In this ordered basis, the restriction of $h_p(\omega)$ is represented by the matrix
\begin{equation}\label{eq:local-matrix}
 h_p(\omega)
 =
 \begin{pmatrix}
 0&0&0&-\overline\omega\\
 0&\left(|p|^2 - \mu \right)&0&0\\
 0&0&\left(|p|^2 - \mu \right)&0\\
 -\omega&0&0&2\left(|p|^2 - \mu \right)
 \end{pmatrix}.
\end{equation}
with lowest eigenvalue given by 
\begin{equation*}
    |p|^2 - \mu\, - \sqrt{\left| |p|^2 - \mu \right|^2 + |\omega|^2}\;. 
\end{equation*}
The lowest eigenvalue of each local block can therefore be computed explicitly. Summing over the momentum sectors yields the following operator bound.
\begin{lemma}[Approximating-Hamiltonian lower bound]
\label{lem:approximating-bound}
For every \(\omega\in\C\), as an operator inequality on the form domain,
\begin{equation}\label{eq:app-operator-bound}
 H_{\rm out}+H_{\rm app}(\omega)
 \ge
 E_{\rm FS}(L,\mu) + \mathscr E^{\rm BCS}_{L,\mu}(\omega).
\end{equation}
Combining this estimate with the square completion gives, for every vector $\varphi$ in the form domain,
\begin{align}
 \langle\varphi,H_{L,\mu}\varphi\rangle
 \ge{}&
 \bigl(E_{\rm FS}(L,\mu) + E^{\rm BCS}_{L,\mu}\bigr)
 \|\varphi\|^2 +
 \nonumber\\
 &-\frac{\sqrt{\mu}L^3}{g}
 \|(Z-\omega)\varphi\|^2.
 \label{eq:app-form-bound}
\end{align}
\end{lemma}

\begin{proof}
Outside the interacting shell, the energy is minimized by filling all modes below the Fermi surface, and hence
\begin{equation}\label{eq:H-out-lower}
 H_{\rm out}
 \ge
 2\sum_{\substack{p\notin\Omega_{L,\mu}\\ |p|^2 < \mu}}\left(|p|^2 - \mu \right).
\end{equation}
Inside the interacting shell, the explicit diagonalization of the local blocks gives
\[
 H_{\rm app}(\omega)
 \ge
 \frac{\sqrt{\mu}L^3}{g}|\omega|^2
 +\sum_{p\in\Omega_{L,\mu}}
 \left(\left(|p|^2 - \mu \right)-\sqrt{\left(|p|^2 - \mu \right)^2+|\omega|^2}\right).
\]
Summing the two above inequalities, and using that for each $p \in \Omega_{L, \mu}$ one can write
\[
 (|p|^2 - \mu)-\sqrt{||p|^2 - \mu|^2+|\omega|^2}
 =
 2(|p|^2 - \mu)\1_{\{|p|^2<\mu\}}
 +||p|^2 - \mu|-\sqrt{||p|^2 - \mu|^2+|\omega|^2}\;,
\]
we find precisely \eqref{eq:app-operator-bound}, i.e.
\begin{equation*}
    H_{\rm out} + H_{\rm app}(\omega) \geq E_{\rm FS}(L, \mu) + \mathscr{E}_{L,\mu}^{\rm BCS}(\omega)\,.
\end{equation*}
The form bound \eqref{eq:app-form-bound} then follows from the square-completion identity \eqref{eq:square-completion} and the fact that the scalar BCS functional is bounded below by its minimum.
\end{proof}
The only remaining difficulty is the negative fluctuation term produced by the square completion. To control it uniformly in the volume, we localize the collective pair field $Z$ near suitable complex values $\omega$.\\
\textbf{Heuristics:} Since pair operators with different momenta commute, the commutator $[Z,Z^*]$ is of order $L^{-3}$. Hence the real and imaginary parts of $Z$ become asymptotically commuting as $L\to\infty$, which suggests that they can be localized simultaneously. If $Z$ is localized within a distance $\ell$ of a complex number $\omega$, the negative fluctuation term produces an error of order $L^3\ell^2$. On the other hand, the IMS localization itself produces an error of order $(L^3\ell^2)^{-1}$. Balancing the two contributions leads to the scale $\ell\sim L^{-3/2}$ and therefore to an error of order one. The remainder of this section makes this argument rigorous.

Since $Z$ is not self-adjoint, we localize its real and imaginary parts separately. We therefore write
\begin{equation}\label{eq:XY}
 X=\frac{Z+Z^*}{2},\qquad
 Y=\frac{Z-Z^*}{2i},\qquad Z=X+iY.
\end{equation}
 The required localization estimates depend only on the commutators of these two self-adjoint operators with each other and with the Hamiltonian. The necessary bounds are collected in the following lemma; their proof is deferred to the appendix \ref{app:commutators}.

\begin{lemma}[Commutator bounds]\label{lem:commutators}
There is a universal constant \(C<\infty\) such that
\begin{align}
 \|[X,Y]\|
 &\le C\left(\frac{g}{\sqrt{\mu}L^3} \right)^2|\Omega_{L, \mu}|,
 \label{eq:XY-commutator}\\
 \|[X,[H_{L,\mu},X]]\|
 +\|[Y,[H_{L,\mu},Y]]\|
 &\le C\left(\left(\frac{g}{\sqrt{\mu}L^3} \right)^2|\Omega_{L, \mu}|+\left(\frac{g}{\sqrt{\mu}L^3} \right)^3|\Omega_{L, \mu}|^2\right).
 \label{eq:double-commutator}
\end{align}
\end{lemma}
We now introduce a smooth partition of unity for the spectral variables of $X$ and $Y$. Let $\chi\in C_c^\infty(\R)$ be real-valued, even and satisfy, for every $t\in\mathbb R$,
\begin{equation}\label{eq:partition-chi}
 \sum_{j\in\Z}\chi^2(t-j)=1\,.
\end{equation}
For a localization scale $\ell>0$ and $j,k\in\Z$, define by functional calculus
\begin{equation}\label{eq:chi-theta}
 \chi_j=\chi\left(\frac{X-j\ell}{\ell}\right),\qquad
 \theta_k=\chi\left(\frac{Y-k\ell}{\ell}\right),\qquad
 A_{jk}=\chi_j\theta_k.
\end{equation}
It then follows from \eqref{eq:partition-chi} that
\begin{equation}\label{eq:partition-A}
 \sum_{j,k\in\Z}A_{jk}^*A_{jk} = 1\,. 
\end{equation}

\begin{lemma}[Two-variable localization]\label{lem:abstract-localization}
Using the above definitions \eqref{eq:chi-theta}, there is a constant $C_\chi<\infty$ (depending only on the function $\chi$) such that
\begin{equation}\label{eq:IMS-general}
 H_{L, \mu}=\sum_{j,k}A_{jk}^*H_{L, \mu}A_{jk}+\mathcal{E}_{\rm loc}
\end{equation}
as quadratic forms, where
\begin{equation}\label{eq:Eloc-bound-general}
 \|\mathcal{E}_{\rm loc}\|
 \le \frac{C_\chi}{\ell^2}
 \bigl(\|[X,[H_{L, \mu},X]]\|+\|[Y,[H_{L, \mu}, Y]]\|\bigr).
\end{equation}
Moreover, for every $\psi \in \mathcal{F}_f$:
\begin{equation}\label{eq:Zlocal-general}
 \sum_{j,k \in \Z}\|(X+iY-(j+ik)\ell)A_{jk}\psi\|^2
 \le C_\chi\left(\ell^2+\frac{\|[X,Y]\|^2}{\ell^2}\right)\|\psi\|^2.
\end{equation}
\end{lemma}

\begin{proof}
We apply an IMS-type identity successively to the two spectral localizations. For a bounded self-adjoint partition of unity satisfying $\sum_j\chi_j^2=1$ and a self-adjoint operator $H$ bounded from below, one has
\[
 \sum_j\chi_j H\chi_j - H
 =
 \frac12\sum_j[\chi_j,[H,\chi_j]]
\]
as a quadratic-form identity. Applying this twice, we find
\begin{align*}
   \sum_{j,k}A_{jk}^*HA_{jk}
 &=\sum_k\theta_k\left(\sum_j\chi_jH\chi_j\right)\theta_k \\
 &=\sum_k\theta_k \left(H + \frac{1}{2}\sum_{j}\left[\chi_j,[H, \chi_j]\right]\right)\theta_k = \\
 &= H + \frac{1}{2}\sum_{k}\left[\theta_k,[H, \theta_k]\right] + \frac{1}{2}\sum_{j,k} \theta_k\left[\chi_j,[H, \chi_j]\right] \theta_k =: H - \mathcal{E}_{loc}\;.
\end{align*}
We first estimate the contribution arising from the localization in the $Y$-variable, namely $\sum_k [\theta_k,[H,\theta_k]]$.
Using the Fourier representation of $\chi$, we write
\[
 \chi(t)=\frac1{\sqrt{2\pi}}
 \int_{\R}\widehat\chi(\xi)e^{i \xi t}\,d\xi \quad\Longrightarrow\quad \theta_k
 =
 \frac1{\sqrt{2\pi}}
 \int_{\R}\widehat\chi(\xi) e^{-i \xi k}e^{i \xi Y/\ell}\,d\xi\;,
\]
so that
\[
\sum_{k \in \Z} [\theta_k,[H,\theta_k]] = \frac{1}{2\pi}\sum_{k \in \mathbb{Z}}\int_{\mathbb{R}\times\mathbb{R}}d\xi\,dw\,e^{-i(\xi + w) k}\,\hat{\chi}(\xi)\hat{\chi}(w)  \left[e^{i\xi Y/\ell}, \left[H, e^{i w Y/\ell}\right]  \right] .
\]

Summing over the localization index is then handled by the Poisson summation formula
\[
 \sum_{k\in\Z}e^{-i(\xi+w)k}
 =
 2\pi\sum_{m\in\Z}\delta(\xi+w-2\pi m)
\]
to get
\[
\sum_{k \in \Z} [\theta_k,[H,\theta_k]] = \sum_{m \in \Z}\int_{\R}dw\,\hat{\chi}(w)\,\hat{\chi}(2\pi m - w)\left[e^{i(2 \pi m - w) Y / \ell}, \left[H, e^{i w Y/ \ell}\right]  \right].
\]
The remaining double commutator is estimated by a Duhamel expansion. Applying the fundamental theorem of calculus to the corresponding unitary conjugation $f(r) := e^{i(1-r)tY/\ell}He^{irtY/\ell}$ yields
\[
 [H,e^{itY/\ell}]
 =
 \frac{it}{\ell}
 \int_0^1
 e^{i(1-r)tY/\ell}[H,Y]e^{irtY/\ell}\,dr\,,
\]
and a second application gives the bound
\[
 \|[e^{isY/\ell},[H,e^{itY/\ell}]]\|
 \le
 \frac{|st|}{\ell^2}\|[Y,[H,Y]]\|\,.
\]

Therefore
\begin{align*}
 \left\|\sum_k[\theta_k,[H,\theta_k]]\right\|
 &\le
 \frac{\|[Y,[H,Y]]\|}{\ell^2}
 \sum_{m\in\Z}\int_{\R}
 |s|\,|2\pi m-s|\,
 |\widehat\chi(s)|\,|\widehat\chi(2\pi m-s)|\,ds\\
 &=\frac{C_{\chi}}{\ell^2}
 \|[Y,[H,Y]]\|.
\end{align*}
The constant $C_\chi$ is finite because $\widehat\chi$ is a Schwartz function. The contribution from the localization in the $X$-variable is estimated in exactly the same way. This proves the bound on $\mathcal E_{\rm loc}$.

It remains to prove the localization estimate \eqref{eq:Zlocal-general}. We treat the $X$- and $Y$-components separately.
Since, by functional calculus, $(X-j\ell)$ commutes with $\chi_j$, we get 
\begin{align*}
 &\sum_{j,k}\|(X-j\ell)\chi_j\theta_k\psi\|^2\\
 &\quad=
 \sum_k
 \left\langle
 \theta_k\psi,
 \left[
 \sum_j\chi_j\,(X-j\ell)^2\,\chi_j
 \right]
 \theta_k\psi
 \right\rangle.
\end{align*}
Introduce the periodic function
\[
 h_\chi(t):=
 \sum_{j\in\mathbb Z}(t-j)^2\chi(t-j)^2 .
\]
The sum is locally finite, so \(h_\chi\) is continuous and in particular
\(
 C_\chi:=\sup_{r\in\mathbb R}h_\chi(r)<\infty.
\)
By functional calculus,
\[
 \sum_j\chi_j(X-j\ell)^2\chi_j
 =
 \ell^2h_\chi(X/\ell)
 \le C_\chi\ell^2.
\]
Consequently,
\begin{equation*}
 \sum_{j,k}\|(X-j\ell)A_{jk}\psi\|^2
 \le 
 C_\chi\ell^2\|\psi\|^2.
 \label{eq:X-local-piece}
\end{equation*}

The $Y$-component requires one additional commutator, because $\chi_j$ is a function of $X$ rather than of $Y$. We write
\[
 (Y-k\ell)\chi_j\theta_k
 =
 \chi_j(Y-k\ell)\theta_k+[Y,\chi_j]\theta_k.
\]
The first term is estimated exactly as in the $X$-component. It therefore remains only to control the commutator term. We claim that
\begin{equation}\label{eq:vector-comm-est}
 \left\|\sum_j[Y,\chi_j]^*[Y,\chi_j]\right\|
 \le \frac{C_\chi}{\ell^2}\|[X,Y]\|^2.
\end{equation}
To prove this estimate, we again use the Fourier representation of $\chi$ and sum over $j$ with the Poisson summation formula. This gives
\[
 \sum_j[Y,\chi_j]^*[Y,\chi_j] = \sum_{m \in \Z}\int_{\R}ds\,\hat{\chi}(s)\overline{\hat{\chi}(s + 2\pi m)}\left[Y, e^{i(2\pi m + s) X/\ell}\right]^* \left[Y, e^{i s X/ \ell}\right]\;.
\]
The same Duhamel argument as above yields
\[
 \|\left[Y,e^{isX/\ell}\right]\|
 \le \frac{|s|}{\ell}\|[X,Y]\|.
\]
Hence
\begin{align*}
 \left\|\sum_j[Y,\chi_j]^*[Y,\chi_j]\right\|
 &\le \frac{\|[X,Y]\|^2}{\ell^2}
 \sum_{m\in\Z}\int_{\R}
 |s|\,| s +2\pi m|\,
 |\widehat\chi(s)|\,|\widehat\chi(s + 2\pi m )|\,ds\\
 &=\frac{C_\chi}{\ell^2}\|[X,Y]\|^2.
\end{align*}
This proves \eqref{eq:vector-comm-est}. Consequently,
\begin{align*}
 \sum_{j,k}\|[Y,\chi_j]\theta_k\psi\|^2
 &=
 \sum_k
 \left\langle
 \theta_k\psi,
 \left(\sum_j[Y,\chi_j]^*[Y,\chi_j]\right)
 \theta_k\psi
 \right\rangle\\
 &\le
 \frac{C_\chi}{\ell^2}\|[X,Y]\|^2
 \sum_k\|\theta_k\psi\|^2\\
 &=
 \frac{C_\chi}{\ell^2}\|[X,Y]\|^2\|\psi\|^2.
\end{align*}
Combining this commutator estimate with the preceding bounds for the $X$- and $Y$-components proves \eqref{eq:Zlocal-general} and completes the proof.
\end{proof}

\section{Proof of the lower bound}

We now combine the approximating-Hamiltonian bound with the localization and commutator estimates obtained in the previous section. Let $\psi\in\mathcal F_f$ be normalized and, for each localization cell, choose the complex parameter $\omega_{jk}:=(j+ik)\ell$. Applying Lemma \ref{lem:approximating-bound} to $A_{jk}\psi$ with this choice of $\omega_{jk}$ and summing over $j,k\in\Z$ gives
\begin{align}
 \langle\psi,H_{L,\mu}\psi\rangle
 &\ge
 \left(E_{\rm FS}(L,\mu)+ E^{\rm BCS}_{L,\mu}\right)
 \underbrace{\sum_{j,k}\|A_{jk}\psi\|^2}_{= \| \psi \|^2 = 1} +
 \nonumber\\
 &\quad
 -\frac{\sqrt{\mu}L^3}{g}
 \sum_{j,k}\|(Z-\omega_{jk})A_{jk}\psi\|^2
 -\|\mathcal{E}_{\rm loc}\|.
 \label{eq:localized-lower}
\end{align}
The first error term measures the deviation of the collective pair field from the center of each localization cell, while the second is the IMS localization error. Applying Lemma \ref{lem:abstract-localization} to these two contributions yields
\begin{align*}
 \langle\psi,H_{L,\mu}\psi\rangle
 &\ge
 E_{\rm FS}(L,\mu)+E^{\rm BCS}_{L,\mu}
 -C\ell^2\frac{\sqrt{\mu}L^3}{g} +
 \nonumber\\
 &\quad
 -\frac{C}{\ell^2}\left( \frac{\sqrt{\mu}L^3}{g}\| [X, Y] \|^2 + \|[X, [H_{L, \mu}, X]] \| + \|[Y, [H_{L, \mu}, Y]] \| \right),
\end{align*}
and using Lemma \ref{lem:commutators} for the commutators gives
\begin{align}
   \langle\psi,H_{L,\mu}\psi\rangle
 &\ge
 E_{\rm FS}(L,\mu)+E^{\rm BCS}_{L,\mu}
 -C\ell^2\frac{\sqrt{\mu}L^3}{g} +
 \nonumber\\
 &\quad
 -\frac{C}{\ell^2}\left[\left(\frac{g}{\sqrt{\mu}L^3} \right)^2|\Omega_{L,\mu}| + \left(\frac{g}{\sqrt{\mu}L^3} \right)^3|\Omega_{L,\mu}|^2 \right] \,.
 \label{eq:error-before-opt}
\end{align}
It remains to control the number of momentum modes contained in the interacting shell. A standard lattice-point estimate gives a constant $C>0$ such that
\begin{equation}\label{eq:N-bound-fixedmu}
 |\Omega_{L,\mu}|\le C \sqrt{\mu} L^3
\end{equation}
for every fixed $\mu>1$ and all sufficiently large $L$.
Inserting this bound into \eqref{eq:error-before-opt} reduces the two commutator contributions to the same volume scale and gives
\begin{equation}\label{eq:error-simplified}
 \langle\psi,H_{L,\mu}\psi\rangle
 \ge
 E_{\rm FS}(L,\mu)+E^{\rm BCS}_{L,\mu}
 -C_{g}
 \left( \ell^2
 \frac{\sqrt{\mu}L^3}{g}
 +\frac{1}{\ell^2}\frac{g}{\sqrt{\mu}L^3}
 \right).
\end{equation}
The two error terms are balanced by choosing the localization scale according to
\begin{equation}\label{eq:ell-opt}
 \ell^2=\frac{g}{\sqrt{\mu}L^3}\,.
\end{equation}
With this choice, both contributions are of order one, uniformly in the volume. We therefore obtain
\begin{equation}\label{eq:lower-final}
 E_0(L,\mu)
 \ge
 E_{\rm FS}(L,\mu)+E^{\rm BCS}_{L,\mu}-C_{\mu,g}.
\end{equation}
Combining this lower bound with the upper bound \eqref{eq:upper-bound} completes the proof of Theorem \ref{thm:main}.

The thermodynamic limit statement of Corollary \ref{cor:thermodynamic} now follows immediately. Indeed, after division by $L^3$, the volume-independent error in Theorem \ref{thm:main} vanishes as $L\to\infty$, while Proposition \ref{prop:scalar-reduction} gives the convergence of the finite-volume BCS energy density to its infinite-volume counterpart.

\appendix
\section{High-density asymptotics}
For the sake of completeness, we sketch the proof of Corollary \ref{cor:high-density}. We start by showing the asymptotic form of $\Delta_{\mu}$ stated in \eqref{eq:gap-high-density}. From Proposition \ref{prop:scalar-reduction}, recall that $\Delta_{\mu}$ solves
\begin{equation*}
    1 = \frac{g}{2\sqrt{\mu}}\int_{\Omega_{\mu}}\frac{1}{\sqrt{||p|^2 - \mu |^2 + \Delta_{\mu}^2}}\frac{dp}{(2\pi)^3}\,.
\end{equation*}
Written in spherical coordinates and doing the change of variables $e := |p|^2 - \mu$, the equation reads
\begin{equation*}\label{eq:gap-radial}
 1
 =
 \frac{g}{8\pi^2}
 \int_{-1}^{1}
 \frac{\sqrt{1+e/\mu}}
 {\sqrt{e^2+\Delta_\mu^2}}\,de\,.
\end{equation*}
Since the denominator is even in \(e\), this can equivalently be
written as
\begin{equation*}\label{eq:gap-symmetrized}
 \frac{8\pi^2}{g}
 =
 \int_0^1
 \frac{
 \sqrt{1+e/\mu}+\sqrt{1-e/\mu}}
 {\sqrt{e^2+\Delta_\mu^2}}\,de.
\end{equation*}

Now using the following simple inequality for $0 \leq x \leq 1$:
\begin{equation*}
    2-x^2 \leq \sqrt{1+x} + \sqrt{1-x} \leq 2
\end{equation*}
and computing the integral, we arrive immediately at
\begin{equation*}
   \frac{4\pi^2}{g}  \leq \operatorname{arcsinh}(1/\Delta_{\mu}) \leq \frac{4\pi^2}{g} + \frac{1}{4\mu^2}\,.
\end{equation*}
From this, it immediately follows that
\begin{equation*}
    \Delta_{\mu} = \underbrace{\frac{1}{\sinh(4\pi^2/g)}}_{:= \Delta_{\infty}} + O(\mu^{-2})\,.
\end{equation*}
The rest of the corollary then follows by computing $e_{\rm BCS}(\mu)$ explictly.

\section{Proof of Lemma \ref{lem:commutators}}\label{app:commutators}
The proof of \ref{lem:commutators} is a simple computation using CAR. In particular, we will use that for any $p, q \in \Lambda_L^*$ one has
\begin{equation*}
    [b_p, b_q^*] = \delta_{p,q}\left( 1 - n_{p, \uparrow} - n_{-p, \downarrow} \right)\,, \qquad n_{p, \sigma} := a^*_{p, \sigma}a_{p, \sigma}\,.
\end{equation*}
Let us start from the commutator of the real and imaginary parts of $Z$, which reads
\[
 [X,Y]
 =\frac{i}{2}\left(\frac{g}{\sqrt{\mu}L^3} \right)^2\sum_{p, q \in \Omega_{L, \mu}}[b_p, b^*_q] = \frac{i}{2}\left(\frac{g}{\sqrt{\mu}L^3} \right)^2\sum_{p\in \Omega_{L, \mu}}(1 - n_{p, \uparrow} - n_{-p, \downarrow})
\]
Thus boundedness of the fermionic creation and annihilation operators immediately gives
\[
\| [X,Y] \| \leq \frac{1}{2}\left(\frac{g}{\sqrt{\mu}L^3} \right)^2|\Omega_{L, \mu}|\,.
\]

To estimate the double commutators, write
\[
 H_{L,\mu}=T+W, \qquad T=\sum_{p\in\Lambda^*}\left(|p|^2 - \mu \right)
 \bigl(n_{p,\uparrow}+n_{-p,\downarrow}\bigr),
 \qquad
 W=-\frac{\sqrt{\mu}L^3}{g}Z^* Z.
\]
Using
\[
 [T,b_p]=-2\left(|p|^2 - \mu \right) b_p,\qquad
 [T,b_p^*]=2\left(|p|^2 - \mu \right) b_p^*,
\]
we obtain
\[
 [T,X] = \frac{g}{\sqrt{\mu}L^3}\sum_{p \in \Omega_{L, \mu}}\left( |p|^2 - \mu \right)(b^*_p - b_p)\;,
\]
and hence
\[
 [X,[T,X]] = \left(\frac{g}{\sqrt{\mu}L^3} \right)^2 \sum_{p \in \Omega_{L, \mu}}\left( |p|^2 - \mu \right) (1 - n_{p, \uparrow} - n_{-p, \downarrow})\;,
\]
Giving immediately
\begin{equation*}
  \|[X,[T,X]]\|\le \left(\frac{g}{\sqrt{\mu}L^3} \right)^2|\Omega_{L, \mu}|.  
\end{equation*}
Similarly one can compute
\begin{align*}
    [W, X] &= \frac{\sqrt{\mu}L^3}{2g}\left(CZ - Z^*C  \right)\,,\\
    C&:= [Z, Z^*] = \left(\frac{g}{\sqrt{\mu}L^3} \right)^2 \sum_{p \in \Omega_{L,\mu}}(1 - n_{p, \uparrow} - n_{-p, \downarrow})\;,
\end{align*}
so that the double commutator with the interaction becomes
\begin{align*}
    [X, [W, X]] &= \frac{\sqrt{\mu}L^3}{2g} \left([X, C]Z - Z^*[X, C] - C^2 \right)\;, \\
    [X,C] &= \frac{1}{2}\left( \frac{g}{\sqrt{\mu}L^3} \right)^3 \sum_{p, q \in \Omega_{L,\mu}}[n_{p, \uparrow} + n_{-p, \downarrow}, b^*_q + b_q] = \\
    &=\left( \frac{g}{\sqrt{\mu}L^3} \right)^3 \sum_{p \in \Omega_{L,\mu}}(b^*_p - b_p)\;,
\end{align*}
and the result simply taking the operator norms for $Z$, $[X,C]$ and $C$. The same holds for $[Y, [H_{L, \mu}, Y]]$.

\section*{Acknowledgments}

We are grateful to Benjamin Schlein for suggesting the specific problem
considered in this work. We further thank Florian Haberberger, Max Duell and Martin Christiansen 
for helpful discussions.

\bibliographystyle{amsalpha}
\bibliography{file_literature_revised_new}

\newcommand{\etalchar}[1]{$^{#1}$}
\providecommand{\bysame}{\leavevmode\hbox to3em{\hrulefill}\thinspace}
\providecommand{\MR}{\relax\ifhmode\unskip\space\fi MR }
\providecommand{\MRhref}[2]{%
  \href{http://www.ams.org/mathscinet-getitem?mr=#1}{#2}
}
\providecommand{\href}[2]{#2}
\begin{thebibliography}{FMRT93}

\bibitem[BCS57]{BCS1957}
J.~Bardeen, L.~N. Cooper, and J.~R. Schrieffer, \emph{Theory of superconductivity}, Phys. Rev. \textbf{108} (1957), no.~5, 1175--1204.

\bibitem[BJ72]{BogolyubovJr1972}
N.~N. Bogolyubov~Jr., \emph{A method for studying model hamiltonians: A minimax principle for problems in statistical physics}, Pergamon Press, Oxford, 1972.

\bibitem[BT93]{BursillThompson1993}
Robert~J. Bursill and Colin~J. Thompson, \emph{Rigorous treatment of the {BCS} model of superconductivity}, Journal of Physics A: Mathematical and General \textbf{26} (1993), no.~4, 769--786.

\bibitem[CFS96]{ChenFroehlichSeifert1996}
Thomas Chen, J{\"u}rg Fr{\"o}hlich, and Maximilian Seifert, \emph{Renormalization group methods: {Landau--Fermi} liquid and {BCS} superconductor}, Fluctuating Geometries in Statistical Mechanics and Field Theory (Fran{\c{c}}ois David, Paul Ginsparg, and Jean Zinn-Justin, eds.), Les Houches Summer School Proceedings, vol.~62, North-Holland, Amsterdam, 1996, pp.~913--970.

\bibitem[CLR88]{CeglaLewisRaggio1988}
W.~Ceg{\l}a, J.~T. Lewis, and G.~A. Raggio, \emph{The free energy of quantum spin systems and large deviations}, Communications in Mathematical Physics \textbf{118} (1988), no.~2, 337--354.

\bibitem[DHM23a]{DeuchertHainzlMaier2023Homogeneous}
Andreas Deuchert, Christian Hainzl, and Marcel~Oliver Maier, \emph{Microscopic derivation of {Ginzburg--Landau} theory and the {BCS} critical temperature shift in a weak homogeneous magnetic field}, Probability and Mathematical Physics \textbf{4} (2023), no.~1, 1--89.

\bibitem[DHM23b]{DeuchertHainzlMaier2023General}
\bysame, \emph{Microscopic derivation of {Ginzburg--Landau} theory and the {BCS} critical temperature shift in general external fields}, Calculus of Variations and Partial Differential Equations \textbf{62} (2023), 203.

\bibitem[DP87]{DuffieldPule1987}
N.~G. Duffield and J.~V. Pul{\'e}, \emph{Thermodynamics of the {BCS} model through large deviations}, Lett. Math. Phys. \textbf{14} (1987), no.~4, 329--331.

\bibitem[DP88]{DuffieldPule1988}
\bysame, \emph{A new method for the thermodynamics of the {B.C.S.} model}, Communications in Mathematical Physics \textbf{118} (1988), 475--494.

\bibitem[FHSS12]{FrankHainzlSeiringerSolovej2012}
Rupert~L. Frank, Christian Hainzl, Robert Seiringer, and Jan~Philip Solovej, \emph{Microscopic derivation of {Ginzburg--Landau} theory}, J. Amer. Math. Soc. \textbf{25} (2012), no.~3, 667--713.

\bibitem[FMRT92]{FeldmanMagnenRivasseauTrubowitz1992}
Joel Feldman, Jacques Magnen, Vincent Rivasseau, and Eugene Trubowitz, \emph{An infinite volume expansion for many fermion {Green}'s functions}, Helvetica Physica Acta \textbf{65} (1992), 679--721.

\bibitem[FMRT93]{FeldmanMagnenRivasseauTrubowitz1993}
\bysame, \emph{An intrinsic {$1/N$} expansion for many-fermion systems}, Europhysics Letters \textbf{24} (1993), no.~6, 437--442.

\bibitem[FT90]{FeldmanTrubowitz1990}
Joel Feldman and Eugene Trubowitz, \emph{Perturbation theory for many fermion systems}, Helvetica Physica Acta \textbf{63} (1990), 156--260.

\bibitem[FT91]{FeldmanTrubowitz1991}
\bysame, \emph{The flow of an electron-phonon system to the superconducting state}, Helvetica Physica Acta \textbf{64} (1991), 213--357.

\bibitem[FW71]{Fetter}
A.~L. Fetter and J.~D. Walecka, \emph{Quantum theory of many-particle systems}, McGraw-Hill, Boston, 1971.

\bibitem[Haa62]{Haag1962}
Rudolf Haag, \emph{The mathematical structure of the {Bardeen--Cooper--Schrieffer} model}, Nuovo Cimento \textbf{25} (1962), no.~2, 287--299.

\bibitem[JBZ{\etalchar{+}}84]{BogolRev}
N.~N.~Bogolyubov (Jr.), I.~G. Brankov, V.~A. Zagrebnov, A.~M. Kurbatov, and N.~S. Tonchev, \emph{Some classes of exactly soluble models of problems in quantum statistical mechanics: the method of the approximating hamiltonian}, Russian Mathematical Surveys \textbf{39} (1984), no.~6, 1--50.

\bibitem[Kat65]{Kato1965}
Yusuke Kato, \emph{Spectrum of the {BCS} reduced hamiltonian in the theory of superconductivity}, Progress of Theoretical Physics \textbf{34} (1965), no.~5, 734--753.

\bibitem[KM67]{Kato1967}
Yusuke Kato and Nobumichi Mugibayashi, \emph{{Friedrichs--Berezin} transformation and its application to the spectral analysis of the {BCS} reduced hamiltonian}, Progress of Theoretical Physics \textbf{38} (1967), no.~4, 813--831.

\bibitem[ML61]{MattisLieb1961}
Daniel~C. Mattis and Elliott~H. Lieb, \emph{Exact wave functions in superconductivity}, Journal of Mathematical Physics \textbf{2} (1961), no.~4, 602--609.

\bibitem[MR13]{martin2013many}
P.A. Martin and F.~Rothen, \emph{Many-body problems and quantum field theory: An introduction}, Theoretical and Mathematical Physics, Springer Berlin Heidelberg, 2013.

\bibitem[PRV89]{PetzRaggioVerbeure1989}
D.~Petz, G.~A. Raggio, and A.~Verbeure, \emph{Asymptotics of {Varadhan}-type and the {Gibbs} variational principle}, Communications in Mathematical Physics \textbf{121} (1989), no.~2, 271--282.

\bibitem[Ric63]{Richardson1963}
R.~W. Richardson, \emph{A restricted class of exact eigenstates of the pairing-force hamiltonian}, Physics Letters \textbf{3} (1963), no.~6, 277--279.

\bibitem[RS64]{RichardsonSherman1964}
R.~W. Richardson and Noah Sherman, \emph{Exact eigenstates of the pairing-force hamiltonian}, Nuclear Physics \textbf{52} (1964), 221--238.

\bibitem[RW89]{RaggioWerner1989BCS}
G.~A. Raggio and R.~F. Werner, \emph{The {Gibbs} variational principle for general {BCS}-type models}, Europhysics Letters \textbf{9} (1989), no.~7, 633--638.

\bibitem[RW91]{RaggioWerner1991}
\bysame, \emph{The {Gibbs} variational principle for inhomogeneous mean-field systems}, Helvetica Physica Acta \textbf{64} (1991), 633--667.

\bibitem[Thi68]{Thirring1968}
W.~Thirring, \emph{On the mathematical structure of the {B.C.S.}-model {II}}, Comm. Math. Phys. \textbf{7} (1968), no.~3, 181--189.

\bibitem[TW67]{ThirringWehrl1967}
W.~Thirring and A.~Wehrl, \emph{On the mathematical structure of the {B.C.S.}-model}, Comm. Math. Phys. \textbf{4} (1967), no.~5, 303--314.

\end{thebibliography}

\end{document}